\documentclass[11pt]{article}
\usepackage{amsmath,amsfonts,amsthm,amssymb,geometry,bbm}
\usepackage{graphics,color}
\usepackage{epsfig}
\usepackage{caption}
\usepackage{algorithmic}
\usepackage[ruled]{algorithm2e}
\usepackage{blkarray}
\usepackage[T1]{fontenc}
\usepackage{enumitem}
\usepackage[numbers]{natbib}

\newtheorem{definition}{Definition}

\newtheorem{lemma}{Lemma}

\newtheorem{theorem}{Theorem}

\usepackage{tabulary}
\usepackage{booktabs}

\newcommand{\M}{\mathcal{M}}
\newcommand{\MT}{\mathcal{MT}}
\newcommand{\T}{\mathcal{T}}
\newcommand{\AP}{\mathcal{AP}}
\newcommand{\br}{\mathbf{r}}
\newcommand{\gen}{\texttt{Gen}\xspace}
\newcommand{\tru}[0]{\mathbf{T}}

\newcommand{\history}{\texttt{H}\xspace}
\newcommand{\indicator}{\mathbbm{1}}

\title{Outsourcing Computation: the Minimal Refereed Mechanism\thanks{The first and third authors gratefully acknowledge the support of the National Science Foundation under CCF \#1618187.  The last two authors gratefully acknowledge the support of NSF under Career Award \#1452915. }} \author{Yuqing Kong\\ Peking University \\ yuqing.kong@pku.edu.cn \and Chris Peikert \\ University of Michigan \\ cpeikert@umich.edu  \and Grant Schoenebeck\\ University of Michigan \\ schoeneb@umich.edu \and Biaoshuai Tao \\ University of Michigan \\ bstao@umich.edu }

\begin{document}
\date{}
\maketitle

\begin{abstract}
We consider a setting where a verifier with limited computation power delegates a resource intensive computation task---which requires a $T\times S$ computation tableau---to two provers where the provers are rational in that each prover maximizes their own payoff---taking into account losses incurred by the cost of computation.
We design a mechanism called the Minimal Refereed Mechanism (MRM) such that if the verifier has $O(\log S + \log T)$ time and $O(\log S + \log T)$ space computation power, then both provers will provide a honest result without the verifier putting any effort to verify the results.  The amount of computation required for the provers (and thus the cost) is a multiplicative $\log S$-factor more than the computation itself, making this schema efficient especially for low-space computations.
\end{abstract}

\section{Introduction}
The growing number of computationally intensive tasks has led to the delegation of computation to ``computing as a service'' platforms such as Amazon's EC2, Microsoft's Azure, etc.  This enables users with widely varying loads to only pay for the computation they need.  This mirrors a larger trend to out-source: Uber (car as a service), Amazon Turk (computer plugged in worker), etc.
When outsourcing tasks, some labors may perform the task honestly due to their intrinsic preference for honesty; however,  often labors need incentives which encourage them to dutifully perform the task. If the requester has ability to (cheaply) verify the  completion of tasks, the incentive problem can be solved naturally by only providing payment for satisfactory results.  However, in the case of outsourcing computation, the verifier cannot necessarily verify the task's completion.  What should an incentive system look like for outsourcing computation?

Motivated by the need for verifiable computations, recent results have drawn on the work of interactive proof systems (IP)---where a resource-limited verifier can verify an extremely complicated proof provided by a untrusted prover---as main ingredients.  However,  the classical IP work diverges from the outsourced computation application in two ways: (1) in the crowdsourcing setting, the cost an honest prover incurs in performing the proof must be taken into consideration, while in the classical IP, the honest prover may suffer a heavy (and uncompensated) burden in proving her result; (2) in the crowdsourcing setting provers can be assumed rational rather than merely untrusted, while the classical IP setting work does not assume or  make use of the rationality of the prover.

Several works (e.g. \cite{azar2012rational,reingold2016constant,canetti2011practical}) either take point (1) or (2) into consideration. But few works consider the both divergences. In this paper, we consider the relation between the related IP work and outsourced computation applications.   We take the effort of the provers into consideration, and provide a mechanism---that we call \emph{the Minimal Refereed Mechanism}--- which harnesses the rationality of provers in that it is individual rational, and has the truthful computation as the only  equilibrium.  In particular, this means that our protocol is robust against agents communicating, as long as they cannot make binding  commitments to one another (for example, to redistributing the payoffs in the future).   Moreover, an honest prover always obtains a positive utility even if her opposite is irrational.  While our mechanism requires that the verifier \emph{can} perform a computation requiring $O(\log S + \log T)$ time and space, in equilibrium, the verifier need only check the equality of answers.

Each prover that faithfully follows our mechanism must spend a factor of $\log(S)$ more computational effort than is required to simply run the computation.  This, of course, must be compensated by the verifier.  However, in the case where the verifier has many different processes to run, we can reduce this overhead by a factor of nearly two.  Instead of having two provers run every program, the vast majority will only be run by one prover.

A key ingredient in the construction of the minimal refereed mechanism is from the ``prisoner dilemma''.  When provers are paid based on whether they have the same output, they may collude to obtain agreement without exerting any effort. To solve the ``collusion'' problem, our minimal refereed mechanism pays an agent who betrays the collusion and \emph{tells the truth} a large reward.  We also draw on techniques form IP so that the verification is possible with dramatically fewer resources than the computation itself requires.

\subsection{Related Work}
\paragraph{Outsourced computation literature}
The most closely related works in this area to the current paper are \citet{belenkiy2008incentivizing, dong2017betrayal}. We all implement the idea of the ``prisoner dilemma'' in outsourced computation. \citet{dong2017betrayal} also employ smart contract to implement the ``prisoner dilemma'' based outsourced computation. However, \citet{belenkiy2008incentivizing,dong2017betrayal} require that the verifier has the ability to run the program by himself and infrequently performs the whole computation to verify the correctness of the prover's output.  Our work only requires that the verifier has the ability to perform a simple $O(\log T + \log S)$ time arbitration process when the provers disagree where the computation size is $T\times S$.

\citet{canetti2011practical} also designs a $O(\log T + \log S)$ time process where a verifier can determine which prover is honest with the help of Merkle hash tree. However, they do not make use of the rationality of the provers but instead assume one of the provers must be honest.

Conceptually, our paper combines the results of the two aforementioned works.  However, naively combining them does not work.  Game theory and computation are notoriously tricky to combine~\cite{halpern2015algorithmic,halpern2016computational}.
For example, our results do not yield a dominant strategy equilibrium as those of \citet{belenkiy2008incentivizing}, and using a collision resistant hash function as in \citet{canetti2011practical} seems not to be enough for our setting.  We carefully integrate the two ideas, and, moreover, provide a delicate game theoretical analysis to show that rational provers must be honest even if the arbitration process gives an arbitrary answer when both of them are dishonest.

\paragraph{Interactive Proof (IP) literature} Since the seminal work of \citet{goldwasser1989knowledge} and \citet{babai1988arthur} introducing interactive proofs (IP), a host of results in closely related models have followed (see, e.g.,~\cite{bitansky2012extractable,goldwasser2008delegating,goldwasser2015delegating,Badrinarayanan:2018:SDL:3188745.3188924}).  In the classical model (e.g.~\citet{lund1992algebraic}, \citet{shamir1992ip}), a verifier with limited computation power has the ability to verify statements provided by a untrustworthy prover with unlimited computation power. This desirable property makes IP work as an important ingredient in many outsourced computation applications. However, in the classical IP work, the verifier usually employs an arithmetization method that imposes a heavy computational burden on the prover even when the prover is honest. Moreover, the classical IP work always considers the worst case---the prover is an adversary.

\citet{azar2012rational} assume that the prover is rational.  With this assumption, \citet{azar2012rational} show that the verifier can easily incentivize a rational prover to provide the answer of $\sharp SAT$ in the following manner: the verifier asks the prover to report $\frac{\sharp SAT}{2^n}$ which can be seen as the prover's prediction for the event that a randomly chosen assignment is satisfied.   The verifier uses a tool called proper scoring rules~\cite{brier1950verification,gneiting2007strictly} to measure the accuracy of the prover's prediction via only one sample and pays the prover the score of the accuracy.  For a $SAT$ instance with $n$ variables, uniformly randomly picking an assignment, the assignment is satisfied with probability $\frac{\sharp SAT}{2^n}$, and so a property of proper scoring rule implies that the prover should provide the exact value of $\frac{\sharp SAT}{2^n}$ to maximize her expected payment. This clever design works with the assumption that the prover can obtain the exact answer without any effort. However, in real life applications, the exponential precision required in some of the reports is very costly to provide.  This influential work has been extended to work for different complexity classes, to improve the efficiency of the verifier, to improve the efficiency of the prover, and to the setting with multiple computation tasks~\cite{azar2013super,campanelli2015sequentially,campanelli2017efficient,campanelli2018fine,guo2014rational,guo2016rational}.  However, while this line of work does explicitly have incentives, the costs of computation are ignored while computing these incentives~(however, see Sect.~3 in~\cite{campanelli2015sequentially}).

Several works successfully design an interactive proof system where the computational resources of both the verifier and honest prover run are limited, but they do not take the rationality of the prover into consideration.   The most closely related work in this area to the current paper are refereed games \cite{feige1997making} and doubly efficient IP \cite{reingold2016constant}. The arbitration process in the current paper is designed based on the idea of the refereed game in \citet{feige1997making}. However, \cite{feige1997making} do not make use of the rationality of provers and still put a heavy burden on the honest prover. \citet{reingold2016constant} designs a doubly efficient and constant-round interactive proofs for languages that have a unique witness---if $x\in L$ there exists a unique witness $y$, of polynomial size, that attests to this.  In their proof system the verifier runs in linear time with respect to $|y|$, and the honest provers run in polynomial time with respect to $|y|$.  However, if we do not have better than polynomial bound of $|y|$ in terms of $|x|$ this tells us little about the required run time of verifier.  Note that the prover's polynomial time bound does not account for the time it takes the prover to find $|y|$, which could be super-polynomial, if $L$ is a hard language (e.g. $L \not\in P$).
\citet{reingold2016constant} do not explicitly take the rationality into consideration.  Moreover, even if we use the prisoner dilemma technique to modify \citet{reingold2016constant} to a mechanism where the verifier does not need to spend effort when the provers are rational, that modified mechanism still requires the verifier has the ability to run a linear time verification (in the size of $|y|$), while our mechanism only requires a sublinear time verification in $|x|$ (as long as the computation itself is computable in subexponential time).

\citet{kalai2019howto} update this work to include various additional settings such as making the computation publicly verifiable (while keeping the result private), non-interactive, and employing standard cryptographic assumptions.
Moreover, their work applies to polynomial computation rather than NP computations.  The running time required by the verifier is still polynomially related ($T^\epsilon$) to the actual running time of the delegated computation $T$, and the verifier incurs an additional polynomial overhead.

The aforementioned \citet{canetti2011practical} uses interactive proofs with multiple provers to design schemes with increased efficiency.
\citet{canetti2013refereed} extends this line of results to be more efficient and apply to more realistic architectures (instead of Turing Machines), but both assume one honest prover.

\citet{gennaro2010non} gives a construction that allows the outsourcing of a single function for multiple inputs,  a different setting that considered here.  The verifier's need for computational power scales  linearly with the output size of the computation by cleverly employing techniques from Yao's garbled circuit and fully-homomorphic encryption.

\citet{teutsch2017scalable} produce a white paper for TrueBit, a system which allows out-sourced computation via smart contracts for the digital currency Ethereum.  The system allows users to post computations with a reward for the answer.  A user proposing to have solved the computation must also post a bounty.  The proposed answer then can be challenged by any user.  A challenge results in an arbitration process where the solver must prove the validity of her solution.  If she fails, she loses her bounty.  If no successful challenge occurs before a deadline, the solver collects the original reward and reclaims her bounty.   Unfortunately, in the equilibrium, agents should shirk the task and report randomly  with a small probability.

\section{Preliminaries}
Consider the scenario where a verifier wants to solve a question $q$ and has program $\mathcal{M}$ that can solve the question.
The verifier, however, only has \emph{limited computational power} and cannot run the code by himself.
Therefore, the verifier gives the program to two agents: Alice and Bob.
In this paper, we design a mechanism for the verifier which collects reports from Alice and Bob, and rewards them based on their reports in a way that incentivizes both agents to faithfully execute the program.
In this section, we first review cryptographical hash functions and the Merkle hash tree which are used by our mechanism, and then discuss the mechanism design goals.

\subsection{Merkle Hash Tree $H(\mathcal{T})$ of Computation Table $\mathcal{T}$}
In our setting, Alice and Bob use the same code $\mathcal{M}$ to solve $q$ if both of them are honest. We assume the program $\mathcal{M}$ requires at most time $T$ and space $S$.

\begin{definition}[computation table]
The \emph{computation table} $\mathcal{T}$ of a Turing machine $\mathcal{M}$ that calculates question $q$ is a $T\times S$ matrix. The first row encodes the input and initial configuration of $\mathcal{M}$. Each row has an \emph{active region} around where the read/write head of $\mathcal{M}$ is located. The last non-blank row has only one non-blank entry---the answer of question $q$.
\end{definition}

\begin{definition}[hash function~\cite{lindell2014introduction}]\label{def:hashFunction}
A \emph{hash function} (with output length $\ell$) is a pair of probabilistic polynomial-time algorithms $(\gen,H)$ satisfying the following:
\begin{itemize}
    \item \gen is a probabilistic algorithm which takes as input a security parameter $1^n$ and outputs a key $k$. We assume that $1^n$ is implicit in $k$.
    \item $H$ takes as input a key $k$ and a string $x\in\{0,1\}^\ast$ and outputs a string $H^k(x)\in\{0,1\}^{\ell(n)}$ (where $n$ is the value of the security parameter implicit in $s$).
\end{itemize}
We call $H^k(x)$ the \emph{hash value} of $x$.
\end{definition}

A standard property that a hash function $(\gen,H)$ has is the \emph{collision-resistance}, meaning that it is computationally infeasible to find a collision---$x,x'\in\{0,1\}^\ast$ such that $H^k(x)=H^k(x')$, even if the algorithm knows the key $k$.

\begin{definition}[collision-resistance]
Given a hash function $(\gen,H)$, an \emph{adversary} $\mathcal{A}$ is a probabilistic algorithm which takes as inputs the security parameter $n$ and a key $k$ generated by $\gen(1^n)$, and outputs $x,x'\in\{0,1\}^\ast$.
The hash function is \emph{collision-resistant} if $\Pr(H^k(x)=H^k(x'))\leq\frac{t^2}{2^n}$ for all probabilistic adversaries $\mathcal{A}$ that run in time at most $t$ (where the randomness in this probability comes from $\mathcal{A}$, not the key generation $k\sim\gen(1^n)$).
\end{definition}

For simplicity, we refer to $H$ or $H^k$ instead of $(\gen,H)$ as a hash function, and all the hash functions in this paper satisfy collision-resistance.
Throughout the paper, we make a standard assumption that the hash function $(\gen,H)$ can only be accessed as a \emph{random oracle}~\cite{bellare1993random}.
As a result, for a fixed message $x$, we assume that the time required to compute the hash value depends only on $n$.
Moreover, it is infeasible to obtain the value $H^k(x)$ without knowing $x$.

\begin{definition}[Merkle (hash) tree]
A Merkle tree is a binary tree in which every internal node stores the hash value of the \emph{concatenation} (denoted by symbol $||$) of its two children and the leaves are the hash values of different data blocks.
\end{definition}

We are interested in constructing a Merkle tree $\MT_H(\T)$ for a computation table $\T$.
Definition~\ref{def:MTT} and Fig.~\ref{fig:tree} illustrates the construction.

\begin{definition}[Merkle tree for a computation table]\label{def:MTT}
Given a computation table $\T$ of a Turing machine $\M$ and a Hash function $H$, the \emph{Merkle tree for $\T$}, denoted by $\MT_H(\T)$, is constructed as follows.
\begin{enumerate}
    \item \emph{The Lower Part of $\MT_H(\T)$}: For each row $\T_i$ of $\T$, we split it into several data blocks of size $\lambda$, and construct a Merkle tree $\MT_H(\T_i)$ where each leaf is the hash value of a data block.
    \item \emph{The Upper Part of $\MT_H(\T)$}: The upper part of $\MT_H(\T)$ is a binary tree with $T$ leaves such that the $i$-th leaf is the root of $\MT_H(\T_i)$. Each internal node has value which is the hash value of the concatenation of its two children as it is in a Merkle tree.
\end{enumerate}
\end{definition}

Throughout the paper, we use $r$ to denote the value of the root node of $\MT_H(\T)$, $r_i$ to denote the value of the root of the subtree $\MT_H(\T_i)$ which corresponding to the $i$-th row of $\T$, and $r_{ij}$ to denote the value of the leaf corresponding to the $j$-th block of the $i$-row.
Denote the $j$-th block of the $i$-row of $\T$ by $b_{ij}$, and then  $r_{ij}$ is the hash value of $b_{ij}$.
We use $r^A,r_i^A,r_{ij}^A,b_{ij}^A$ to refer to the corresponding values that Alice provides (which may be subjected to Alice's strategical manipulation), and let $r^B,r_i^B,r_{ij}^B,b_{ij}^B$ have similar meanings.
We sometimes abuse the notations a little bit and use $r,r_i,r_{ij}$ to refer to the nodes themselves instead of the values stored in these nodes.

An advantage of the Merkle tree is that we can verify the consistency between a single data block $b_{ij}^A$ (or $b_{ij}^B$) and the value $r^A$ (or $r^B$) with time complexity only $O(\log T+\log S)$, as we will see soon.

\begin{definition}[consistent path]
Given $\MT_H(\T)$ for a computation table $\T$ and two nodes $u,v$ of $\MT_H(\T)$ such that $u$ is an ancestor of $v$, we say \emph{the path from $u$ to $v$ is consistent} if for each node $w$ on the (shortest) $u$-$v$ path the value of $w$ is the hash value of the concatenation of the values stored in $w$'s children.
In particular, for $w=v$ being a leaf, the value must be the hash value of the corresponding data block.
\end{definition}

From the definition above, to check the consistency of a path from $u$ to $v$, we need the values of all the nodes on the path and all their children.
For example, to check the consistency between $b_{ij}^A$ and $r^A$, we need to check if the path from the root to the leaf corresponding to this data block is consistent.
Since this path has length $O(\log(ST))=O(\log S+\log T)$ and each node on the path has at most two children, this consistency can be checked in time $O(\log S+\log T)$.

\begin{figure}[!ht]
  \caption{Hash Computation Tableau via Merkle Tree}
  \centering
    \includegraphics[width=0.8\textwidth]{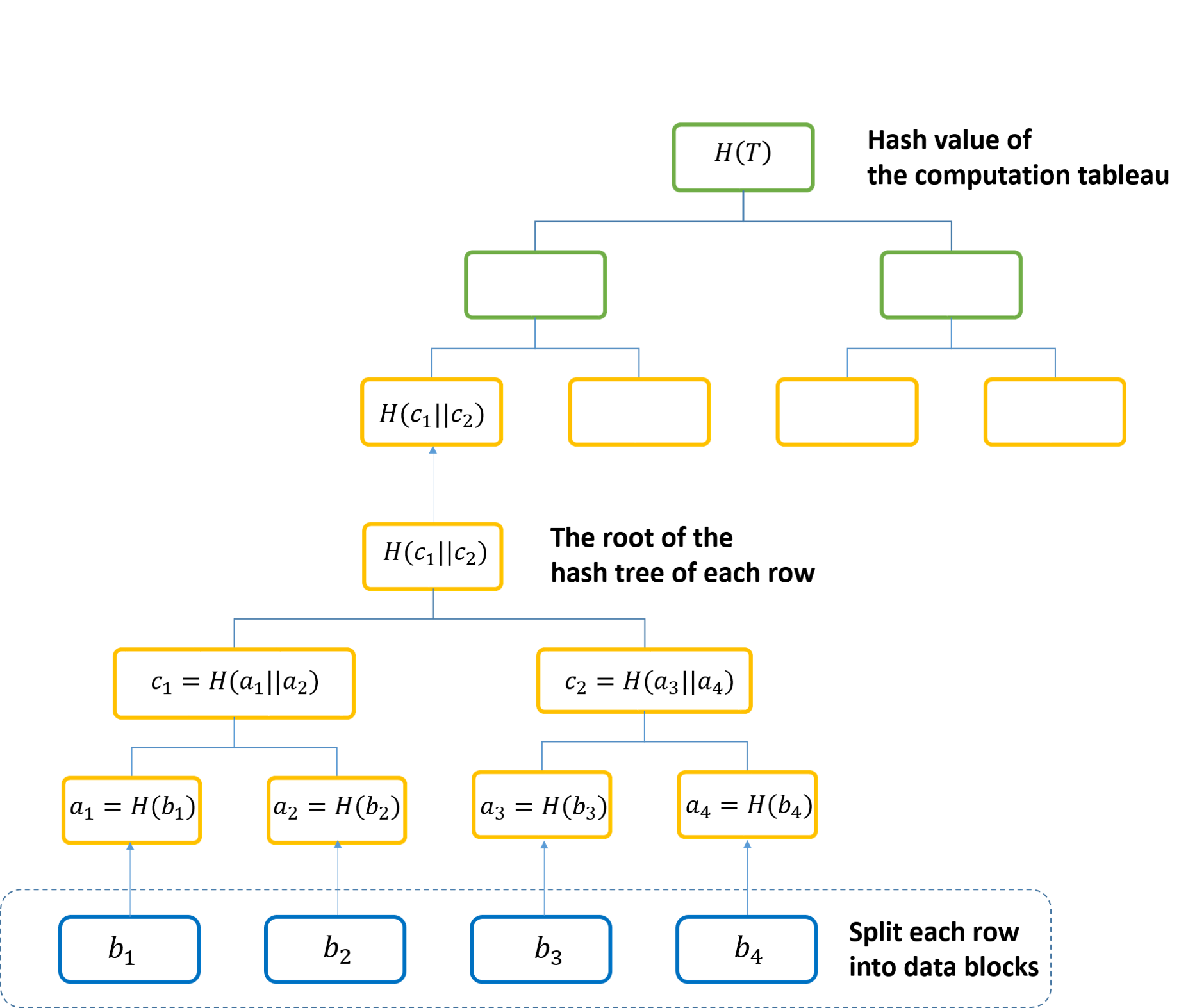}
    \label{fig:tree}
\end{figure}

\subsection{An Informal Description of Mechanism Design Goals}
\label{sect:goals_informal}
In this section, we describe some basic game theory concepts as well as our mechanism design goals in an informal way.
Formal definitions of those are deferred to Sect.~\ref{sect:goal}.

Given a game, we define the \emph{utility} $\mu$ of each agent as the payment she receives minus the amount of effort she spends.
For each agent, we define the \emph{pure (mixed) strategy} $s$ as a mapping such that $s$ maps every stage of the game and every information set the agent might receive in that stage to an action (the distribution over all possible actions) in the next stage. In other words, $s$ tells an agent what to do for every possible situation throughout the game and determines the action (the distribution over all possible actions) the agent will take at any stage of the game. We call $(s^A,s^B)$ a \emph{strategy profile} of Alice and Bob where $s^A$ is Alice's strategy and $s^B$ is Bob's strategy.

A strategy profile $\mathbf{s}=(s^A,s^B)$ is a \emph{Nash equilibrium} if neither Alice nor Bob can deviate to another strategy to obtain a strictly higher utility when $\mathbf{s}$ is mutually known by Alice and Bob.
We say a mechanism is \emph{truthful} if the strategy profile where each of the agents plays a truthful strategy that always submits the correct answer to the verifier is a Nash equilibrium.
We say a mechanism is \emph{strongly truthful} if the truth-telling strategy profile is the only Nash equilibrium in some sense.

We have two goals for our mechanism design.
Other than the \emph{strong truthfulness}, our second goal is that, through an iterative query process, the mechanism must be able to verify the correctness of the answer that Alice and Bob provide in logarithmic time: $O(\log S+\log T)$.

\section{Minimal Refereed Mechanism}

Remember that the high level idea is from \emph{the prisoner's dilemma}. When Alice and Bob are paid based on whether they have the same output, they may collude to obtain agreement without exerting any effort. To solve this ``collusion'' problem, our \emph{minimal refereed mechanism} pays an agent who betrays the collusion and \emph{tells the truth} a large reward.  We also draw on techniques from IP so that the verification is possible with dramatically fewer resources than the computation itself requires.

\paragraph{Minimal Refereed Mechanism ($\mathcal{M}$, $f_H(\cdot)$, $\mathcal{AP}$, $d_1$, $d_2$):}
\begin{description}
\item[Step 1] The verifier samples a hash key $k\sim\gen(1^n)$ and assigns a program $\mathcal{M}$ to both Alice and Bob and asks them to commit to $\mathbf{r}=f_H(\mathcal{M})$. We call $f_H(\cdot)$ the \emph{commitment function}, which is a mapping from a program $\mathcal{M}$ to a report profile $(a,t,r)$, where $a$ is the output of $\mathcal{M}$, $t$ is the time spent in computing $\mathcal{M}$ (i.e., the number of non-blank rows in $\T$), and $r$ is the value of the root of the Merkle hash tree $\MT_H(\mathcal{T})$, for $\T$ being the computation table of $\M$.
\item[Step 2 (Computation Stage)] Alice and Bob do the computation separately and commit $\mathbf{r}^A$ and $\mathbf{r}^B$ to the verifier privately.
\item[Step 3 (Arbitration Stage)] If $\mathbf{r}^A=\mathbf{r}^B$, the verifier pays both Alice and Bob the amount $d_1(t)$ that depends on and is monotone in $t$, the time spent as reported by both agents. Otherwise, the verifier runs an \emph{arbitration process} in which the verifier asks both agents several questions and finally announces for each of Alice and Bob if she(he) is a \emph{winner}. The verifier pays each winner $d_2\gg d_1(T)$ and each loser $0$ (recall that $T$ is the number of rows in $\T$ including blank rows, which is an upper bound on $t$).
\end{description}

\paragraph{MRM: Arbitration Process $\mathcal{AP}$}
The arbitration process shown in Algorithm~\ref{alg:AP} takes in the two commitments $\br^A,\br^B$ that are different, and outputs either ``winner'' or ``loser'' for each of Alice and Bob. Below we give a verbal summary of Algorithm~\ref{alg:AP}.

When Alice and Bob agree with the value of the root of $\MT_H(\T)$, then they must disagree on either $a$ or $t$.
The verifier first checks, in the case $t^A\neq t^B$, if the last rows from both agents contains the halting state, and if the agent reporting the larger running time has a halting state in the middle row $\min\{t^A,t^B\}$ (Line~4 to Line~8).
Notice that in a correct execution of $\M$, the halting state should appear and only appear in the last non-blank row.
The arbitration process terminates immediately if an agent is caught for violating this, and moves on otherwise.

The verifier then asks their values for the root of the subtree corresponding to the $\min\{t^A,t^B\}$-th row (Line~9):
\begin{itemize}
    \item If they agree with each other, the verifier checks the path from the root of this subtree to the first block of the $\min \{t^A,t^B\}$-th row (Line~11), and announces the winner or the loser based on if the agent can provide a consistent path (Line~12).
    \item If they disagree, then the verifier checks the consistency of the path from the root of this subtree to the root of the entire tree $\MT_H(\T)$ (Line~14), and announces the winner or the loser based on the consistency (Line~15).
\end{itemize}

When Alice and Bob disagree on the value of the root of $\MT_H(\T)$, the verifier runs a subroutine \texttt{FirstDivergence}($r^A$, $r^B$, $r$) defined in Algorithm~\ref{alg:firstDiv} to figure out \emph{the first place on which Alice disagrees with Bob in $\mathcal{T}$} (Line~18).
The subroutine \texttt{FirstDivergence} takes in three inputs: a hash value $v^A$, a hash value $v^B$, and a node $v$ in the Merkle tree, where $v^A$ and $v^B$ are the values Alice and Bob provide (respectively) for the node $v$, and $v^A\neq v^B$.
It outputs either the identity of the agents (either one or two) which are identified as ``liars'', or a block $b_{ij}$ in $\T$.
As a brief description, \texttt{FirstDivergence} travels from the node $v$ to a leaf based on the following rules: at a node $u$ with children $u_1,u_2$ during the traversal, \texttt{FirstDivergence} checks if the hash is consistent, i.e., if $H(u_1^A||u_2^A)=u^A$ and $H(u_1^B||u_2^B)=u^B$, and moves on to the left-most child $u_i$ with $u_i^A\neq u_i^B$.
If an inconsistency in the hash computation is found during the traversal, \texttt{FirstDivergence} terminates and outputs the identity of the agent(s) with the inconsistent computation as the liar(s).
Otherwise, the traversal process will not end until a leaf is reached.
This is because we start at $v$ on whose value Alice and Bob disagree, and Alice and Bob must disagree on at least one of the children $u_1,u_2$ if they disagree on the parent $u$.
When the traversal ends at a leaf $r_{ij}$, assuming both agents have not broken the hash function, we know that $b_{ij}$ is the first block where Alice and Bob disagree.
Intuitively, \texttt{FirstDivergence} performs a binary search, viewing the Merkle tree as the binary search tree, and checks the consistency of hashing during the search.

If \texttt{FirstDivergence}($r^A$, $r^B$, $r$) outputs the identity of the liar(s), the algorithm terminates by announcing the liar(s) being the loser(s) (Line~20).
If the output is a block $b_{ij}$, then we consider two cases: $i=1$ and $i>1$.
\begin{itemize}
\item  If $i=1$, i.e., the block is in the first row, since the first row of $\T$ encodes the input of $\M$, it is easy for the verifier to check the correctness of $b_{ij}^A$ and $b_{ij}^B$ by herself (Line~24 and Line~25).
\item  If $i>1$, i.e., the block is in a middle row, the verifier asks Alice and Bob the value of the corresponding block $b_{(i-1)j}$ in the previous row which contains the active region (Line~27).
If Alice and Bob agree on $b_{(i-1)j}$, the verifier calculates $b_{ij}$ by himself and spot the liar(s) who has a different value than the verifier (Line~29 and Line~30).
If Alice and Bob disagree on $b_{(i-1)j}$, the verifier checks the consistency of the path from the leaf $r_{(i-1)j}$ all the way to the root $r$ of the Merkle tree (Line~32 and Line~33). \end{itemize}

\begin{algorithm}[!ht]
  \caption{The arbitration process $\AP$}
  \label{alg:AP}
  \begin{algorithmic}[1]
    \STATE\textbf{Input}: $\br^A=(a^A,t^A,r^A)$ and $\br^B=(a^B,t^B,r^B)$ with $\br^A\neq\br^B$
    \STATE\textbf{Output}: $o^A,o^B\in\{1,0\}$ where $1$ stands for ``winner'' and $0$ stands for ``loser''
    \STATE\textbf{if} $r^A=r^B$:
    \STATE\hspace{0.3cm}\textbf{if} $t^A\neq t^B$, assume without loss of generality $t^A>t^B$:
    \STATE\hspace{0.3cm}\hspace{0.3cm}request values $b_{t^A,1}^A,b_{t^B,1}^A$ from Alice and value $b_{t^B,1}^B$ from Bob\footnotemark
    \STATE\hspace{0.3cm}\hspace{0.3cm}\textbf{if} $b_{t^B,1}^B$ does not contain a halting state: \textbf{return} $o^A=1$ and $o^B=0$
    \STATE\hspace{0.3cm}\hspace{0.3cm}\textbf{if} $b_{t^B,1}^A$ contains a halting state or $b_{t^A,1}^A$ does not: \textbf{return} $o^A=0$ and $o^B=1$
    \STATE\hspace{0.3cm}\textbf{endif}
    \STATE\hspace{0.3cm}let $t=\min\{t^A,t^B\}$, and request the values $r_t^A,r_t^B$ from both agents
    \STATE\hspace{0.3cm}\textbf{if} $r_t^A=r_t^B$:
    \STATE\hspace{0.3cm}\hspace{0.3cm}request the relevant values and check if the two paths $r_{t1}^A$-$r_t^A$ and $r_{t1}^B$-$r_t^B$ are consistent
    \STATE\hspace{0.3cm}\hspace{0.3cm}\textbf{return} $o^A=\indicator(r_{t1}^A$-$r_t^A\mbox{ is consistent})$, $o^B=\indicator(r_{t1}^B$-$r_t^B\mbox{ is consistent})$
    \STATE\hspace{0.3cm}\textbf{else}
    \STATE\hspace{0.3cm}\hspace{0.3cm}request the relevant values, and check if the two paths $r_{t}^A$-$r^A$ and $r_{t}^B$-$r^B$ are consistent
    \STATE\hspace{0.3cm}\hspace{0.3cm}\textbf{return} $o^A=\indicator(r_{t}^A$-$r^A\mbox{ is consistent})$, $o^B=\indicator(r_{t}^B$-$r^B\mbox{ is consistent})$
    \STATE\hspace{0.3cm}\textbf{endif}
    \STATE\textbf{else}
    \STATE\hspace{0.3cm}implement \texttt{FirstDivergence}($r^A,r^B,r$)
    \STATE\hspace{0.3cm}\textbf{if} the output is the identity of one or two agents:
    \STATE\hspace{0.3cm}\hspace{0.3cm}\textbf{return} $o^A=\indicator(A\mbox{ is not an output})$ and $o^B=\indicator(B\mbox{ is not an output})$
    \STATE\hspace{0.3cm}\textbf{else}
    \STATE\hspace{0.3cm}\hspace{0.3cm}let $b_{ij}$ be output of FirstDivergence($r^A,r^B,r$) (we know $b_{ij}^A\neq b_{ij}^B$)
    \STATE\hspace{0.3cm}\hspace{0.3cm}\textbf{if} $i=1$:
    \STATE\hspace{0.3cm}\hspace{0.3cm}\hspace{0.3cm}request $b_{ij}^A$ and $b_{ij}^B$, and check their correctness\footnotemark
    \STATE\hspace{0.3cm}\hspace{0.3cm}\hspace{0.3cm}\textbf{return} $o^A=\indicator(b_{ij}^A\mbox{ is correct})$ and $o^B=\indicator(b_{ij}^B\mbox{ is correct})$
    \STATE\hspace{0.3cm}\hspace{0.3cm}\textbf{else}
    \STATE\hspace{0.3cm}\hspace{0.3cm}\hspace{0.3cm}request $b_{(i-1)j}^A$ and $b_{(i-1)j}^B$ where $j$-th block in row $(i-1)$ contains the active region\footnotemark
    \STATE\hspace{0.3cm}\hspace{0.3cm}\hspace{0.3cm}\textbf{if} $b_{(i-1)j}^A=b_{(i-1)j}^B$:
    \STATE\hspace{1.2cm}compute $b_{ij}$ from $b_{(i-1)j}^A=b_{(i-1)j}^B$, and check the correctness of $b_{ij}^A$ and $b_{ij}^B$
    \STATE\hspace{1.2cm}\textbf{return} $o^A=\indicator(b_{ij}^A\mbox{ is correct})$ and $o^B=\indicator(b_{ij}^B\mbox{ is correct})$
    \STATE\hspace{0.3cm}\hspace{0.3cm}\hspace{0.3cm}\textbf{else}
    \STATE\hspace{1.2cm}request the relevant values and check the consistency of $r_{(i-1)j}^A$-$r^A$ and $r_{(i-1)j}^B$-$r^B$
    \STATE\hspace{1.2cm}\textbf{return} $o^A=\indicator(r_{(i-1)j}^A$-$r^A\mbox{ is consistent})$ and $o^B=\indicator(r_{(i-1)j}^B$-$r^B\mbox{ is consistent})$
    \STATE\hspace{0.3cm}\hspace{0.3cm}\hspace{0.3cm}\textbf{endif}
    \STATE\hspace{0.3cm}\hspace{0.3cm}\textbf{endif}
    \STATE\hspace{0.3cm}\textbf{endif}
    \STATE\textbf{endif}
  \end{algorithmic}
\end{algorithm}

\begin{algorithm}[!ht]
  \caption{Function \texttt{FirstDivergence}($v^A$, $v^B$, $v$) in $\AP$}
  \label{alg:firstDiv}
  \begin{algorithmic}[1]
    \STATE\textbf{Input}: a node $v$, two hash values $v^A,v^B$ with $v^A\neq v^B$
    \STATE\textbf{Output}: a set of liars $\ell\subseteq\{A,B\}$, \emph{or} a block $b_{ij}$
    \STATE\textbf{if} $v:=r_{ij}$ is a leaf:
    \STATE\hspace{0.3cm}request $b_{ij}^A$ and $b_{ij}^B$
    \STATE\hspace{0.3cm}\textbf{if} $\exists X\in\{A,B\}:H(b_{ij}^X)\neq r_{ij}^X$: \textbf{return} $\ell$ which includes all such $X$
    \STATE\hspace{0.3cm}\textbf{return} $b_{ij}$
    \STATE\textbf{else}
    \STATE\hspace{0.3cm}let $v_1,v_2$ be the children of $v$, and request $v_1^A,v_2^A,v_1^B,v_2^B$ \footnotemark
    \STATE\hspace{0.3cm}\textbf{if} $\exists X\in\{A,B\}:H(v_1^X||v_2^X)\neq v^X$: \textbf{return} $\ell$ which includes all such $X$
    \STATE\hspace{0.3cm}let $\displaystyle i=\min_{i'\in\{1,2\}:v_{i'}^A\neq v_{i'}^B}i'$, and implement \texttt{FirstDivergence}($v_i^A$, $v_i^B$, $v_i$)
    \STATE\textbf{endif}
  \end{algorithmic}
\end{algorithm}

\paragraph{Complexity analysis of $f_H$ and $\mathcal{AP}$:}
The time complexity of computing $f_H$ is $O(S+T\log S+T)$ since for every $i$-th, $(i+1)$-th row, given the hash tree of the $i$-th row, we only need to modify the path from the active region to the root ($O(\log S)$) to obtain the hash tree of $(i+1)$-th row. The time complexity of $\mathcal{AP}$ is $O(\log S+\log T)$. The main computations are the computation of \texttt{FirstDivergece} and the consistency check of paths. Both of them require $\log S+\log T$ time. The verifier needs $O(\log S+\log T)$ space to record the position of the leaf (which represents the path from the root to that leaf). Thus, a verifier, who has the ability to run a computation that needs $O(\log S+\log T)$ time and $O(\log S+\log T)$ space, can run an MRM with $\mathcal{AP}$ as the arbitration process.
This shows the achievement of the second goal mentioned in the last section.

\section{Mechanism Design Goals Revisited}
\label{sect:goal}
Under the mechanism MRM$(\M,f_H(\cdot),\AP,d_1,d_2)$, Alice and Bob are playing a game in \emph{the extensive-form}.
A strategy $s$ should specify an action at each information set.
To make the notion of Nash equilibrium make sense, we need to define a \emph{strategy space} which includes all the possible strategies.
This seemly easy task turns out to be tricky.
For example, at an information set in the middle of the arbitration process where the agent must supply the hash value of certain node, $s$ should specify this value so that the agent can report it to the mechanism.
A natural definition for the space of all possible actions here is the space of all valid hash values: $\{0,1\}^{\ell(n)}$.
However, this is problematic: we know that in this space some of the values make the hashing consistent, and we need to account for the cost/effort that agent spends to find those values instead of just ``selecting'' one of them as an action.
In this section, we aim to provide rigorous definitions for those game theory terms such as strategy space, utility function, equilibrium concepts, as well as our notions of truthfulness and strong truthfulness.

\addtocounter{footnote}{-4}
\stepcounter{footnote}\footnotetext{We assume that only the first block in the output row (the last non-blank row) of $\T$ is non-blank, as the output of a program has a small size in most cases. We can, for example, announce an agent as a loser if he provides an over-sized output.}
\stepcounter{footnote}\footnotetext{The first row encodes the input of the program $\M$, so the verifier can easily check the correctness.}
\stepcounter{footnote}\footnotetext{Occasionally, the active region is at the boundary of the two blocks $b_{(i-1)j}$ and $b_{(i-1)(j\pm1)}$. In this case we request two blocks from each agent, and the remaining part of the algorithm stays the same.}

We start by defining a \emph{strategy}.
Before this, we adopt a binary string representation of the position of a node in a binary tree: the root node is represented by the empty string $\emptyset$; a string of length $\ell$ represent a node at level $\ell+1$ such that, in the unique shortest path connecting this node and the root, the $i$-th edge in the up-to-down order connects a node to its left child if the $i$-th bit of this string is $0$, and it connects a node to its right child if the $i$-th bit is $1$.
When dealing with a Merkle tree, we treat a data block $b_{ij}$ as the unique child of $r_{ij}$ (so that those $r_{ij}$'s are no longer leaves, and the leaves are those $b_{ij}$'s).
For example, the node $c_2$ in the tree shown in Fig.~\ref{fig:tree} is represented by the string $001$, and the data block $b_2$ is represented by $00010$.
For a binary string $p$, we denote by $\pi(p)$ the node in a binary tree that $p$ represents.

\begin{definition}[strategy]
A \emph{pure strategy} $s:=(s_c,s_{ap})$ is a collection of two deterministic Turing machines $s_c,s_{ap}$ where
\begin{itemize}
    \item $s_c(\M,H^k)$ takes in a program $\M$ and a hash function $H^k$ with key $k$ as inputs, and outputs a triple $(a,t,h)$ where $a$ is a string, $t$ is a positive integer, and $h$ is a binary string of length $\ell(n)$ (where $\ell(\cdot)$ is the parameter in the hash function, see Definition~\ref{def:hashFunction}), and
    \item $s_{ap}(\M,H^k,p,\history)$ takes as inputs a program $\M$, a hash function $H^k$ with key $k$, a binary string $p$, and a string $\history$, and outputs a binary string $h$,
\end{itemize}
such that
\begin{itemize}
    \item when the agent is asked to commit $\br$ in the computation stage of MRM, the strategy $s$ specifies $\br=s_c(\M,H^k)$, and
    \item when, in the arbitration stage, the agent is request for the hash value of a node (which may possibly be a data block) in $\MT_H(\T)$ represented by $\pi(p)$, the strategy $s$ specifies $s_{ap}(\M,H^k,p,\history)$, where the string $\history$ records all the historical queries in $\AP$.
\end{itemize}
\stepcounter{footnote}\footnotetext{Both Alice and Bob are supposed to construct a Merkle tree for all the $T$ rows in $\M$ including those blank rows after the output row, so the two trees they construct should have the same structure. Therefore, when $v^A$ has children, $v^B$ should also have. If this is not the case, we can directly output the agent with incorrect tree structure as a liar.}
A \emph{mixed strategy} is a probability distribution over the space of all possible pure strategies as it is in standard game theory.\footnote{Alternatively, we can view a mixed strategy as at least one of $s_c,s_{ap}$ being a probabilistic Turing machine.}
\end{definition}

\begin{definition}[utility and effort]\label{def:utility}
Given a strategy profile $(s^A=(s_c^A,s_{ap}^A),s^B=(s_c^B,s_{ap}^B))$, the \emph{utility} of Alice $\mu^A(s^A,s^B)$ is the expected payment she receives minus the expected \emph{effort} she spent on implementing the two Turing mechanisms, where the \emph{effort} is the total number of the state transitions in the implementations of $s_c^A,s_{ap}^A$ when answering the queries throughout the execution of MRM (in particular, $s_c^A$ is implemented exactly once, while the number of implementations of $s_{ap}^A$ depends on both $s^A$ and $s^B$), and the expectation is taken over the randomness in both the hash key sampling and the strategies that are possibly mixed.
The utility for Bob, $\mu^B(s^A,s^B)$, is defined similarly.
\end{definition}

Once we have defined strategies and utilities, the definition of Nash equilibrium becomes standard, so we omit it here.

The natural way to define the truthful strategy is to require that the strategy always specifies the correct answers in both the computation stage and the arbitration stage.
However, under such definition, it is impossible that the truth-telling profile is the only Nash equilibrium (or MRM is strongly truthful according to our definition in Sect.~\ref{sect:goals_informal}).
For example, it will also be a Nash equilibrium if Bob plays a truthful strategy and Alice plays a strategy that specifies the correct answers for the computation stage but incorrect answers at a certain information set in the arbitration stage  which she believes will never be reached.
However, as the verifier's ultimate goal is to know the output $a$ of the program $\M$, what we really need is that both agents are honest in the computation stage only.
We discuss in Sect.~\ref{sect:otherNotions} about other common notions, such as dominant strategy truthfulness and truth-telling as a subgame perfect Nash equilibrium, and why we do not use them.

\begin{definition}[truthful strategy]
A strategy $s=(s_c,s_{ap})$ is \emph{truthful} if $s_c(\M,H^k)=f_{H^k}(\M)$.
We use $\mathbf{T}$ to denote the set of all truthful strategies.
We also define a specific truthful strategy $\tau=(\tau_c,\tau_{ap})\in\mathbf{T}$ called \emph{the absolutely truthful strategy}, which, in addition, satisfies that $\tau_{ap}(\M,H^k,p,\history)$ outputs the correct hash value of the node $\pi(p)$ in $\MT_H(\T)$ for every $p$ such that $\pi(p)\in\MT_H(\T)$.
\end{definition}

\begin{definition}[(strongly) truthful MRM]
The mechanism MRM is \emph{truthful} if $(\tau,\tau)$ is a Nash equilibrium, and is \emph{strongly truthful} if, in addition, $(s^A,s^B)$ being a Nash equilibrium implies $s^A,s^B\in\mathbf{T}$.
\end{definition}

Throughout the paper, we use $M(i)$ to denote the (minimum) amount of effort required to compute the first $i$ rows of $\T$ and the Merkle tree $\MT_H(\T)$, $M_c:=M(T)$ to denote the maximum effort a truthful agent can spend in the computation stage, and $M_{ap}$ to denote the maximum possible effort a truthful agent can spend in the arbitration stage.
In particular, for the absolutely truthful strategy $\tau=(\tau_c,\tau_{ap})$, the effort of computing $\tau_c(\M,H)=(a,t,h)$ is $M(t)$, and we use $M_c^\tau$ to denote this.
From our discussion earlier, we have $M_c=O(S+T\log S+T)$ and $M_{ap}=O(\log S + \log T)$.

\section{Strong Truthfulness of MRM}
\label{sect:main}
In this section, we prove that the mechanism MRM is strongly truthful with appropriate choices of parameters.

\begin{theorem}\label{thm:main}
For any program $\mathcal{M}$, if we choose large enough security parameter $n$ such that $n>2\log(M_c+M_{ap})+C$ for some constant $C$, there exists $b$ such that MRM ($\mathcal{M}$, $f_H$, $\mathcal{AP}$, $d_1$, $d_2$) with $d_2=2(M_c+M_{ap})+2b$ and $d_1(t)=M(t)+b$ is strongly truthful, where $t$ is the time spent in computing $\M$ as reported by both agents (their reports are the same if payment $d_1(\cdot)$ is considered).
\end{theorem}

We will prove this main theorem by showing the following four lemmas.

\begin{lemma}\label{lem:epsilon-informative}
Given a program $\M$, a hash function $(\gen,H)$ and a pure strategy $s=(s_c,s_{ap})\notin\tru$ such that the total effort of computing $s_c(\M,H^k)$ is strictly less than $M_c^\tau$, there exists $\epsilon<1$ such that
$\Pr_{k\sim\gen(1^n)}\left(s_c(\M,H^k)=f_{H^k}(\M)\right)<\epsilon.$
\end{lemma}
\begin{proof}
This follows immediately from our assumption that $H^k$ is a random oracle.
\end{proof}

\begin{lemma}\label{lem:delta-zeta-truthful}
An agent playing the absolutely truthful strategy $\tau$ always wins in the arbitration process $\AP$, regardless of the strategy the other agent plays.
Moreover, when the security parameter $n$ is large enough with $n>2\log(M_c+M_{ap})+15$, if there exists a dishonest agent that plays a pure strategy $s\notin\mathbf{T}$ and all dishonest agents spend effort at most $\zeta:=2^{(n-15)/2}$, the probability that $\AP$ announces two winners is smaller than $\delta:=2^{-10}$.
\end{lemma}
\begin{proof}
Notice that $\AP$ only checks the validity of the halting state, the consistency of certain selected paths in $\MT_H(\T)$ and the correctness of a small part of computation (from $b_{(i-1)j}$ to $b_{ij}$), and only announce a loser if there is any mistake in these.
The honest agent playing $\tau$ will always win in $\AP$.

To show the second part, notice that the maximum effort an honest agent spends is at most $M+M_{ap}$ which is strictly less than $\zeta$.
Assuming dishonest agents spend effort at most $\zeta$, the total effort of $s_A,s_B$ in $\AP$ is at most $2\zeta<2^{(n-10)/2}$.
It is enough to show that, if at least one of $s_A,s_B$ is not in $\tru$, the event that both agents win in $\AP$ implies a collision in the hash function $H^k$.

We assume for the sake of contradiction that both agents win in $\AP$.
When $\AP$ is implemented, we know that $\br^A\neq\br^B$.
We consider three different cases: 1) $r^A\neq r^B$, 2) $r^A=r^B$ but $t^A\neq t^B$, and 3) $r^A=r^B$, $t^A=t^B$, but $a^A\neq a^B$.

In the first case, the subroutine \texttt{FirstDivergence} in Line~18 of Algorithm~\ref{alg:AP} is implemented, and the consistency of the path from the root $r$ to a leaf $r_{ij}$ is checked.
The nature of \texttt{FirstDivergence} ensures $r_{ij}^A\neq r_{ij}^B$.
Firstly, we have $i>1$, as Alice and Bob, both being winners, should agree on the first row of $\T$.
For $i>1$, if both agents also disagree on $r_{(i-1)j}$, moving from $r_{(i-1)j}$ to the root along the shortest path, they will eventually agree on the value of a middle node on the path (otherwise \texttt{FirstDivergence} will never end at $b_{ij}$), this implies a collision.
If they agree on $r_{(i-1)j}$ but disagree on $b_{(i-1)j}$, then again a collision is found.
If the agree on both $r_{(i-1)j}$ and $b_{(i-1)j}$, then at least one of Alice and Bob cannot make the transition $b_{(i-1)j}\mapsto b_{ij}$ correct, and will be announced as a loser at Line~30, which contradicts to our assumption that both agents win.

In the second case, both agents will be checked against the halting state validity.
Let $t=\min\{t^A,t^B\}$ and assume $t^A>t^B$. We know that $b_{t1}^A\neq b_{t1}^B$ if both agents survive from Line~4 to Line~8, as $b_{t1}^A$ should not contain a halting state while $b_{t1}^B$ should.
If both agents agree on the subtree root $r_t$, then a collision can be found on the path $r_{t1}$-$r_t$.
If both agents disagree on $r_t$, since in this case they agree on $r$, a collision can be found on the path $r_t$-$r$.

The third case is similar. Since $a^A\neq a^B$, we have $b_{t1}^A\neq b_{t1}^B$.
We can find a collision on either $r_{t1}$-$r_t$ or $r_t$-$r$.
\end{proof}

Note that Lemma~\ref{lem:delta-zeta-truthful} does not say anything about the situation where both two agents are dishonest and only one of them wins $\mathcal{AP}$.

\begin{lemma}
\label{lem:truthful}
MRM($\M$, $f_H$, $\mathcal{AP}$, $d_1$, $d_2$) is truthful if $d_1(t)=M(t)+b$, $b>\delta d_2+\frac{\epsilon M_c}{1-\epsilon}$ and $\zeta>d_2$.
\end{lemma}
\begin{proof}
When Bob plays $\tau$, we need to show that $\mu^A(\tau,\tau)\geq\mu^A(s^A,\tau)$ for any $s^A$, and we can assume without loss of generality that $s^A$ is a pure strategy.
If $s^A\in\mathbf{T}$, then $s_c^A=\tau_c$.
We know that the arbitration $\AP$ will not be implemented, so Alice receives the same payment and spends the same effort as she will be when she plays $\tau$.

If $s^A\notin\mathbf{T}$, then the best Alice can hope is that either she has a good guess for the commitment (with probability bounded by $\epsilon$ by Lemma~\ref{lem:epsilon-informative}) or she is lucky in $\mathcal{AP}$ and wins together with Bob with small effort (with probability bounded by $\delta$ by Lemma~\ref{lem:delta-zeta-truthful}) or she wins together with Bob via $\zeta$ effort.
Therefore, $\mu^A(s^A,\tau)\leq\epsilon(M_c+b)+(1-\epsilon)\left(\max\{\delta d_2+(1-\delta)0,d_2-\zeta\}\right)$.
Simple calculations reveal that the assumptions $b>\delta d_2+\frac{\epsilon M_c}{1-\epsilon}$ and $\zeta>d_2$ imply $\epsilon(M_c+b)+(1-\epsilon)\delta d_2< b$ and $\epsilon(M_c+b)+(1-\epsilon)(d_2-\zeta)<b$.
Thus, $\mu^A(s^A,\tau)< b=d_1(t)-M(t)=\mu^A(\tau,\tau)$.
\end{proof}

\begin{lemma}
\label{lem:focaltruthful}
MRM($\M$, $f_H$, $\mathcal{AP}$, $d_1$, $d_2$) is strongly truthful if $d_1(t)=M(t)+b$, $d_2=2(M_c+M_{ap})+2b$ for $\frac{\zeta-2(M_c+M_{ap})}{2}>b>\frac{2\delta (M_c+M_{ap})+\frac{\epsilon M_c}{1-\epsilon}}{1-2\delta}>0.$
\end{lemma}
\begin{proof}
It is easy to see $\zeta>d_2$. By some calculations, we also have $b>\delta d_2+\frac{\epsilon M_c}{1-\epsilon}$.
Thus, we know MRM is truthful by Lemma~\ref{lem:truthful}.

To show MRM is strongly truthful, we consider any (pure or mixed) strategy Nash equilibrium $(s^A,s^B)$. We will show that $(s^A,s^B)$ is a Nash equilibrium if and only if both agents tell the truth in the computation stage: $s^A,s^B\in\mathbf{T}$.

We classify the possible outcomes of $(s^A,s^B)$ into the below disjoint cases:
\begin{description}[noitemsep,topsep=0pt,parsep=0pt,partopsep=0pt]
\item[$O$] Alice and Bob agree with each other and both of them tell the truth.
\item[$A$] Alice and Bob agree with each other on correct commitment but at least one of them is dishonest.
\begin{description}
\item[$A_1$] Alice is truthful but Bob spends effort less than the effort of $\tau_c(\M,H)$.
\item[$A_2$] Bob is truthful but Alice spends effort less than the effort of $\tau_c(\M,H)$.
\item[$A_3$] Both Alice and Bob spend effort less than the effort of $\tau_c(\M,H)$.
\end{description}
\item[$B$] Alice and Bob agree with each other on a wrong commitment.
\item[$C$] Alice wins in $\mathcal{AP}$ via a non-truthful strategy. Bob loses.
\item[$D$] Alice wins in $\mathcal{AP}$ via a truth-telling strategy. Bob loses.
\item[$E$] Bob wins in $\mathcal{AP}$ via a non-truthful strategy. Alice loses.
\item[$F$] Bob wins in $\mathcal{AP}$ via a truth-telling strategy. Alice loses.
\item[$G$] Both lose in $\mathcal{AP}$.
\item[$H$] Both win in $\mathcal{AP}$.
\begin{description}
\item [$H_1$] Alice is honest in the computation stage. Bob is dishonest.
\item [$H_2$] Bob is honest in the computation stage. Alice is dishonest.
\item [$H_3$] Both of them are dishonest in the computation stage.
\end{description}
\end{description}

In the remaining part of this proof, when we mentioned probability of certain event, the randomness is from both the hash key generation $k\sim\gen(1^n)$ and the mixed strategy.
We will show that $(s^A,s^B)$ \emph{is a Nash equilibrium if and only if the probability that outcome $O$ happens is 1.}
The same arguments in the proof of Lemma~\ref{lem:truthful} show the if direction, so we focus on the only-if direction.

Let $M^*=M+M_{ap}$ be the maximum effort the truth-telling strategy can cost in the whole MRM game.
For convenience, we write $\Pr(A,B)$ as the probability event $A$ or event $B$ happens. Since the cases we consider are disjoint, we can see $\Pr(A,B)=\Pr(A)+\Pr(B)$.

First of all, any pure strategies that spend effort at least $\zeta$ is strictly dominated, since the maximum possible payment $d_2$ is less than $\zeta$.
Therefore, in the remaining part of this proof, we assume with probability $0$ that one of Alice and Bob will spend effort at least $\zeta$.

We compare Alice's expected utility when Alice plays $\tau$ and Bob plays $s_B$ with the expected utility of $(s^A,s^B)$.
Notice that a truthful agent always has utility $M(t)+b-M(t)=b$ if $\AP$ is not launched, and has utility at least $d_2-M^*$ if $\AP$ is launched.
\begin{align*}
    & \mu^A(s^A,s^B)-\mu^A(\tau,s^B)\\
    \leq& \Pr(A_2,A_3)\cdot (M_c+b-b) +\Pr(B)\cdot (M_c+b-(d_2-M^*)) \\
    &   +\Pr(C)\cdot (d_2-(d_2-M^*)) + \Pr(E)\cdot (0-(d_2-M^*)) +\Pr(F)\cdot (0-b)\\
    &  + \Pr(G)\cdot (0-(d_2-M^*)) + \Pr(H_2) \cdot(d_2-b)+\Pr(H_3) \cdot(d_2-(\epsilon b+(1-\epsilon)(d_2-M^*))
\end{align*}
Since $(s^A,s^B)$ is a Nash equilibrium, $ \mu^A(s^A,s^B)-\mu^A(\tau,s^B)\geq 0$.
By simplification, rearranging terms and substituting $d_2-(\epsilon b+(1-\epsilon)(d_2-M^*))<d_2-b$ (which is straightforward to show) for the coefficient of $\Pr(H_3)$, we have
$$\Pr(A_2,A_3)\cdot M_c+\Pr(C)\cdot M^*+\Pr(H_2,H_3) \cdot(d_2-b)\geq$$
\begin{equation}\label{eqn:1}
    \qquad \Pr(B)\cdot (d_2-M_c-b-M^*)+\Pr(E,G)\cdot(d_2-M^*)+\Pr(F)\cdot b.
\end{equation}

Moreover, let $\Sigma$ be the event that Alice spends effort less than the effort of $\tau_c(\M,H)$, we know $\Pr(A_2,A_3)=\Pr((A_2,A_3)\land \Sigma)=\Pr(\Sigma)\Pr(A_2,A_3\mid\Sigma)\leq \epsilon \Pr(\Sigma)$ by Lemma~\ref{lem:epsilon-informative}, which implies
$\Pr(B,C,E,F,G,H_2,H_3)\geq (1-\epsilon)\Pr(\Sigma)$, and which further implies
\begin{equation}\label{eq:A}
    \Pr(A_2,A_3)\leq \frac{\epsilon}{1-\epsilon} \Pr(B,C,E,F,G,H_2,H_3)
    =\frac{\epsilon}{1-\epsilon}\left(\Pr(B,C,E,F,G)+\Pr(H_2,H_3)\right).
\end{equation}
Let $\Pi$ be the event that Alice is dishonest in the computation stage and $\AP$ is implemented. We know $\Pr(H_2,H_3)=\Pr((H_2,H_3)\land \Pi)=\Pr(\Pi)\Pr(H_2,H_3\mid\Pi)\leq\delta\Pr(\Pi)$ by Lemma~\ref{lem:delta-zeta-truthful} (as mentioned earlier, we can assume no one spends at least $\zeta$ effort), which implies
$\Pr(C,E,F,G)\geq (1-\delta)\Pr(\Pi)$, and which further implies
\begin{align}\label{eq:H}
\Pr(H_2,H_3)\leq \frac{\delta}{1-\delta} \Pr(C,E,F,G).
\end{align}

After replacements according to (\ref{eq:A}) and (\ref{eq:H}), we can rewrite (\ref{eqn:1}) as
\begin{align} \label{lab:key}
& \Pr(C)\left(M^*+\frac{\epsilon\cdot M_c}{1-\epsilon}+\frac{\delta\cdot (d_2-b+\frac{\epsilon\cdot M_c}{1-\epsilon})}{1-\delta}\right)\nonumber\\
\geq & \Pr(B)\left((d_2-M_c-b-M^*)-\frac{\epsilon\cdot M_c}{1-\epsilon}\right)+\Pr(E,G)\left(d_2-M^*-\frac{\epsilon\cdot M_c}{1-\epsilon}-\frac{\delta\cdot (d_2-b+\frac{\epsilon\cdot M_c}{1-\epsilon})}{1-\delta}\right)\nonumber\\
& +\Pr(F)\left(b-\frac{\epsilon\cdot M_c}{1-\epsilon}-\frac{\delta\cdot (d_2-b+\frac{\epsilon\cdot M_c}{1-\epsilon})}{1-\delta}\right).
\end{align}

Since $b>\delta d_2+\frac{\epsilon M_c}{1-\epsilon}$, we have the coefficient of $\Pr(F)$ in (\ref{lab:key}) satisfies
$
b-\frac{\epsilon\cdot M_c}{1-\epsilon}-\frac{\delta\cdot (d_2-b+\frac{\epsilon\cdot M_c}{1-\epsilon})}{1-\delta}=\frac{1}{1-\delta}\left(b-\delta d_2-\frac{\epsilon\cdot M_c}{1-\epsilon}\right)>0.
$
This further implies
\begin{align}\label{eq:b}
d_2=2M^*+2b>2M^*+2\frac{\epsilon\cdot M_c}{1-\epsilon}+2\frac{\delta\cdot (d_2-b+\frac{\epsilon\cdot M_c}{1-\epsilon})}{1-\delta},
\end{align}
so the coefficients of $\Pr(G)$ in (\ref{lab:key}) is positive.
By $b>\delta d_2+\frac{\epsilon M_c}{1-\epsilon}$ again and $d_2=2M^*+2b$,
$(d_2-M_c-b-M^*)-\frac{\epsilon\cdot M_c}{1-\epsilon}=b+M^*-M_c-\frac{\epsilon\cdot M_c}{1-\epsilon}>  \delta d_2+M^*-M_c=\delta d_2+M_{ap}>0,$
so the coefficients of $\Pr(B)$ in (\ref{lab:key}) is also positive.
Therefore, we have
\begin{align}\label{lab:3}
\Pr(C)\left(M^*+\frac{\epsilon M_c}{1-\epsilon}+\frac{\delta (d_2-b+\frac{\epsilon\cdot M_c}{1-\epsilon})}{1-\delta}\right)\geq \Pr(E)\cdot\left(d_2-M^*-\frac{\epsilon M_c}{1-\epsilon}-\frac{\delta (d_2-b+\frac{\epsilon\cdot M_c}{1-\epsilon})}{1-\delta}\right).
\end{align}

Symmetrically, by analyzing Bob, we have
\begin{align}\label{lab:4}
\Pr(E)\left(M^*+\frac{\epsilon M_c}{1-\epsilon}+\frac{\delta (d_2-b+\frac{\epsilon\cdot M_c}{1-\epsilon})}{1-\delta}\right)\geq \Pr(C)\cdot\left(d_2-M^*-\frac{\epsilon M_c}{1-\epsilon}-\frac{\delta (d_2-b+\frac{\epsilon\cdot M_c}{1-\epsilon})}{1-\delta}\right).
\end{align}

Equation (\ref{eq:b}) implies the coefficient of $\Pr(E)$ is strictly greater (less) than the coefficient of $\Pr(C)$ in (\ref{lab:3}) (in (\ref{lab:4})), then $\Pr(C)=\Pr(E)=0$ for otherwise (\ref{lab:3}) and (\ref{lab:4}) cannot be valid at the same time.
When $\Pr(C)=\Pr(E)=0$, (\ref{lab:key}) implies that $\Pr(B)=\Pr(F)=\Pr(G)=0$,
which, by (\ref{eq:A}) and (\ref{eq:H}), further implies that $\Pr(A_2,A_3)=0$ and $\Pr(H_2,H_3)=0$. Combining with a similar analysis for Bob, we will have $\Pr(O)=1$ in every pure or mixed equilibrium.
\end{proof}

\begin{proof}[Proof for Theorem~\ref{thm:main}]
Equipped with Lemma~\ref{lem:focaltruthful}, we only need to show there exists $b$ such that  $\frac{\zeta-2(M_c+M_{ap})}{2}>b>\frac{2\delta (M_c+M_{ap})+\frac{\epsilon M_c}{1-\epsilon}}{1-2\delta}>0$.
The assumption $n>2\log(M_c+M_{ap})+C$ and the facts $\zeta=2^{(n-15)/2},\delta<1,\epsilon<1$ make sure $b$'s existence.
\end{proof}

We remark that as $\delta$ and $\epsilon$ go to 0, the surplus payment over the cost of the computation $b$, can be arbitrarily small.

\section{On Other Notions of Truthfulness}
\label{sect:otherNotions}
The mechanism MRM is \emph{not dominant-strategy truthful} (meaning $\tau$ or any $s\in\tru$ is a dominant strategy), unlike the prisoner's dilemma game.
If fixing Bob's strategy such that Bob only computes $\M$ up to the $i$-th row and reports it as $a$, then $\tau$ (or any other truthful strategies) is not a best respond to Alice, as Alice only needs to compute $\M$ up to the $(i+1)$-th row (maybe also manually inserts a halting state in the $(i+2)$-th row) and wins in $\AP$.

Notice also that $(\tau,\tau)$ is not a \emph{subgame perfect Nash equilibrium} (SPNE).
For our problem, SPNE is difficult to achieve, but it is also unnecessary and unnatural in some sense.
Consider that both agents reach an information set in the arbitration process where they have already been asked the first block of the $i$-th row, $b_{ij}$, in the computation tableau and they have both agreed on an incorrect value of $b_{ij}$.
The mechanism MRM then checks the first block of the $(i+1)$-th row, $b_{i(j+1)}$, to see if both of them agree (if not, the mechanism will compute the Turing machine transition from the $b_{ij}$ to $b_{i(j+1)}$ to verify the correctness).
In this case, both agents should report the value of $b_{i(j+1)}$ computed from $b_{ij}$, instead of the ``correct'' value of $b_{i(j+1)}$ (meaning the value $b_{i(j+1)}$ if the program $\mathcal{M}$ was executed correctly).
An agent who ``keeps lying'' by computing $b_{i(j+1)}$ based on the incorrect $b_{ij}$ will win if the other agent behaves truthfully.

However, the above situation will never happen for rational agents. Our mechanism is ``truthful'' in a natural sense that Nash equilibria are induced at all the information sets along the truthful equilibrium path $(\tau,\tau)$.
The equilibrium $(\tau,\tau)$, even though failing to be a SPNE, does not rely on non-credible threats.

Our notion of strong truthfulness ensures that the mechanism MRM will obtain the correct output of $\mathcal{M}$ if the two agents are rational (by playing any equilibrium strategy).

\bibliographystyle{plainnat}
\bibliography{reference}

\end{document}